\PassOptionsToPackage{dvipsnames}{xcolor}
\documentclass[copyright,creativecommons]{eptcs}
\providecommand{\event}{ACT 2019} 

\usepackage{booktabs}   
\usepackage{subcaption} 
\usepackage[utf8]{inputenc}
\usepackage[english]{babel}
\usepackage[T1]{fontenc}
\usepackage{amsmath}
\usepackage{hyperref}
\usepackage{tikz}
\usepackage{amsthm}
\usepackage{amssymb}
\usepackage{stmaryrd}

\newtheorem{theorem}{Theorem}
\newtheorem{lemma}[theorem]{Lemma}
\newtheorem{proposition}[theorem]{Proposition}

\newtheorem{corollary}[theorem]{Corollary}

\newtheorem{definition}[theorem]{Definition}
\newtheorem{assumption}[theorem]{Assumption}
\newtheorem{example}[theorem]{Example}
\newtheorem{nonexample}[theorem]{Non-Example}

\newtheorem{remark}[theorem]{Remark}

\newcommand{\cupp}[1]{\ensuremath{\overset{\smallsmile}{#1}}}%

\newcommand{\id}{\text{id}}
\newcommand{\Id}{\text{Id}}

\newcommand{\op}{\ensuremath{\mathrm{op}}}

\newcommand{\VV}{\ensuremath{\mathbf{V}}}

\newcommand{\BBE}{\ensuremath{\mathcal{B}}}

\newcommand{\CC}{\ensuremath{\mathbf{C}}}

\newcommand{\CCE}{\ensuremath{\mathcal{C}}}

\newcommand{\EE}{\ensuremath{\mathcal{E}}}
\newcommand{\FE}{\ensuremath{\mathcal{F}}}
\newcommand{\GE}{\ensuremath{\mathcal{G}}}
\newcommand{\TE}{\ensuremath{\mathcal{T}}}
\newcommand{\HE}{\ensuremath{\mathcal{H}}}
\newcommand{\KE}{\ensuremath{\mathcal{K}}}

\newcommand{\DD}{\ensuremath{\mathbf{D}}}
\newcommand{\DDE}{\ensuremath{\mathcal{D}}}

\newcommand{\DCPO}{\ensuremath{\mathbf{CPO}}}

\newcommand{\dcpo}{\DCPO}
\newcommand{\dcpobs}{\ensuremath{\mathbf{CPO}_{\perp!}}}

\newcommand{\cpobs}{\dcpobs}
\newcommand{\cpo}{\dcpo}

\newcommand{\M}{\ensuremath{\mathbf{M}}}

\newcommand{\Set}{\ensuremath{\mathbf{Set}}}

\newcommand{\Hilb}{\ensuremath{\mathbf{Hilb}_{\lambda}^{\leq 1}}}

\newcommand{\lrb}[1]{{\llbracket #1 \rrbracket}}

\usetikzlibrary{decorations.pathmorphing}
\usetikzlibrary{decorations.markings}
\usetikzlibrary{decorations.pathreplacing}
\usetikzlibrary{arrows}
\usetikzlibrary{shapes}

\pgfdeclarelayer{edgelayer}
\pgfdeclarelayer{nodelayer}
\pgfsetlayers{edgelayer,nodelayer,main}

\tikzstyle{braceedge}=[decorate,decoration={brace,amplitude=10pt}]
\tikzstyle{square box}=[rectangle,fill=white,draw=black,minimum height=6mm,minimum width=6mm,yshift=0.7mm]
\tikzstyle{wire label}=[font=\footnotesize, auto,swap]

\tikzstyle{none}=[inner sep=0pt]
\tikzstyle{gn}=[circle,fill=Lime,draw=Black,line width=0.8 pt]
\tikzstyle{rn}=[circle,fill=Red,draw=Black, line width=0.8 pt]
\tikzstyle{H}=[rectangle,fill=Yellow,draw=Black]
\tikzstyle{line}=[scalar,fill=White,draw=Black]
\tikzstyle{io}=[rectangle,fill=White,draw=Black]
\tikzstyle{block}=[rectangle,fill=Orange,draw=Black]
\tikzstyle{graph}=[circle,fill=White,draw=Black]
\tikzstyle{empty}=[rectangle,fill=none,draw=none]
\tikzstyle{scaled}=[rectangle,fill=none,draw=none, font=\small]
\tikzstyle{box}=[rectangle,fill=White,draw=Black]
\tikzstyle{dot}=[circle,fill=Black,draw=Black,inner sep=0pt,minimum size=1pt]
\tikzstyle{small dot}=[circle,fill=Black,draw=Black,inner sep=0pt,minimum size=1pt]
\tikzstyle{Dot}=[circle,fill=Black,draw=Black,inner sep=0pt,minimum size=3pt]
\tikzstyle{diam}=[rectangle,fill=Black,draw,yscale=1.2,rotate=45]
\tikzstyle{gangle}=[rectangle,fill=Lime,draw=Black]
\tikzstyle{rangle}=[rectangle,fill=Red,draw=Black]
\tikzstyle{circ}=[circle,fill=none,draw=Black,scale=1.3]
\tikzstyle{ellip}=[ellipse,fill=none,draw=Black,scale=1.3,minimum width =1.3cm]
\tikzstyle{ellip2}=[ellipse,fill=White,draw=Black,scale=1.3,minimum width =3cm]
\tikzstyle{bbox}=[rectangle,fill=Blue,draw=Blue,scale=0.6]
\tikzstyle{gg}=[shape=rectangle,fill=White,draw=Black,dashed]

\tikzstyle{white circle}=[circle,fill=none,draw=Black,scale=1]
\tikzstyle{black circle}=[circle,fill=Black,draw=Black,scale=1]
\tikzstyle{grey circle}=[circle,fill=Gray,draw=Black,scale=1]
\tikzstyle{white rectangle}=[rectangle,fill=none,draw=Black,scale=1]

\tikzstyle{nodev}=[circle,fill=none,draw=Black,scale=1]
\tikzstyle{greynode}=[circle,fill=Grey,draw=Black,scale=1]
\tikzstyle{blacknode}=[circle,fill=Black,draw=Black,scale=1]
\tikzstyle{wirev}=[circle,fill=Black,draw=Black,inner sep=0pt,minimum size=3pt]
\tikzstyle{wirevred}=[circle,fill=Red,draw=Black,inner sep=0pt,minimum size=3pt]

\tikzstyle{simple}=[-,draw=Black]
\tikzstyle{to}=[->,draw=Black]
\tikzstyle{naturalto}=[-{Implies},double distance=1.5pt]
\tikzstyle{bdirected}=[<->,draw=Black]
\tikzstyle{bothdirs}=[bdirected,draw=Black]
\tikzstyle{bothdirsred}=[bdirected,draw=Red]
\tikzstyle{blue}=[-,draw=Blue]
\tikzstyle{redd}=[directed,draw=Red]
\tikzstyle{redu}=[-,draw=Red]
\tikzstyle{blued}=[directed,draw=Blue]
\tikzstyle{dash}=[dashed,draw=Black]
\tikzstyle{ddash}=[->,dashed,draw=Black]
\tikzstyle{dashedd}=[->,dashed]
\tikzstyle{dashedred}=[dashed,draw=Red]

\tikzstyle{equal-arrow}=[double equal sign distance]

\tikzstyle{dotpic}=[scale=0.5]

\tikzstyle{every picture}=[baseline=-0.25em]

\newcommand{\stikz}[2][1]{\scalebox{#1}{
\InputIfFileExists{#2}{}{\input{./tikz/#2}}
}}
\newcommand{\cstikz}[2][1]{\begin{center}\stikz[#1]{#2}\end{center}}

\title{Reflecting Algebraically Compact Functors}
\date{}
\author{Vladimir Zamdzhiev
  \institute{Universit\'e de Lorraine, CNRS, Inria, LORIA, F 54000 Nancy, France}
}
\def\titlerunning{Reflecting Algebraically Compact Functors}
\def\authorrunning{V. Zamdzhiev}

\begin{document}
\maketitle

\begin{abstract}
A \emph{compact} $T$-algebra is an initial $T$-algebra whose inverse is a
final $T$-coalgebra. Functors with this property are said to be
\emph{algebraically compact}. This is a very strong property used in
programming semantics which allows one to interpret recursive datatypes
involving mixed-variance functors, such as function space. The construction of
compact algebras is usually done in categories with a zero object where some form
of a limit-colimit coincidence exists. In this paper we consider a more
abstract approach and show how one can construct compact algebras in categories
which have neither a zero object, nor a (standard) limit-colimit coincidence
by reflecting the compact algebras from categories which have both.  In doing so,
we provide a \emph{constructive} description of a large class of algebraically
compact functors (satisfying a compositionality principle) and show our methods
compare quite favorably to other approaches from the literature.
\end{abstract}

\section{Introduction}\label{sec:intro}

\emph{Inductive datatypes} for programming languages can be used to represent
important data structures such as lists, trees, natural numbers and many
others. When providing a denotational interpretation for such languages, type
expressions correspond to functors and one has to be able to construct their
initial algebras in order to model inductive datatypes~\cite{lehman-smyth}. If
the admissible datatype expressions allow only pairing and sum types, then the
functors induced by these expressions are all polynomial functors, i.e.,
functors constructed using only coproducts and (tensor) product connectives,
and the required initial algebra may usually be constructed using Adámek's
celebrated theorem~\cite{adamek-original}.

However, if one also allows function types as part of the admissible datatype
expressions, then we talk about \emph{recursive datatypes} and their
denotational interpretation requires additional structure. A solution advocated
by Freyd~\cite{freyd} and Fiore and Plotkin~\cite{fiore-plotkin} is based on
\emph{algebraically compact functors}, i.e., functors $F$ which have an initial
$F$-algebra whose inverse is a final $F$-coalgebra. $F$-algebras with this
property are called \emph{compact} within this paper.

The celebrated limit-colimit coincidence theorem~\cite{smyth-plotkin} and other similar theorems are
usually used for the construction of compact algebras with starting point a zero
object of the category where the language is interpreted.  However, if one is
interested in semantics for mixed linear/non-linear lambda calculi, then it
becomes necessary to also solve recursive domain equations within categories
that do not have a zero object.

In this paper, we demonstrate how one can construct compact algebras in categories
which do not have a zero object and we do so without (explicitly) assuming the
existence of any limits or colimits whatsoever. Our methods are based on
\emph{enriched} category theory and we show how this allows us to reflect compact
algebras from categories with strong algebraic compactness properties into
categories without such properties. The results which we present are also
compositional and this allows us to provide constructive descriptions of
large classes of algebraically compact functors using formal grammars.

\section{A Reflection Theorem for Algebraically Compact Functors}\label{sec:reflect}

In this section we show how initial algebras, final coalgebras and compact
algebras may be reflected.

\begin{definition}
Given an endofunctor $T: \CC \to \CC$, a \emph{$T$-algebra} is a pair $(A, a),$
where $A$ is an object of $\CC$ and $TA \xrightarrow{a} A$ is a morphism of $\CC$.
A $T$-algebra morphism $f : (A, a) \to (B, b)$ is a morphism $f: A \to B$
of $\CC$, such that the following diagram:
\cstikz{t-algebra.tikz}
commutes. The dual notion is called a $T$-coalgebra.
\end{definition}

Obviously, $T$-(co)algebras form a category. A $T$-(co)algebra is initial (final) if it
is initial (final) in that category.

\begin{definition}
An endofunctor $T: \CC \to \CC$ is (1) \emph{algebraically complete} if it has
an initial $T$-algebra (2) \emph{algebraically cocomplete} if it has a final
$T$-coalgebra and (3) \emph{algebraically compact} if it has an initial
$T$-algebra $T \Omega \xrightarrow \omega \Omega,$ such that $T \Omega
\xleftarrow{ \omega^{-1} } \Omega$ is a final $T$-coalgebra.
\end{definition}

Next, we recall a lemma first observed by Peter Freyd.

\begin{lemma}[{\cite[pp. 100]{freyd}}]
Let $\CC$ and $\DD$ be categories and $F: \CC \to \DD$ and $G : \DD \to \CC$ functors. If $GF \Omega \xrightarrow{ \omega } \Omega$ is an initial $GF$-algebra, then $FGF \Omega \xrightarrow{F \omega} F \Omega$ is an initial $FG$-algebra.
\end{lemma}

By dualising the above lemma, we obtain the next one.

\begin{lemma}\label{lem:coalgebras}
Let $\CC$ and $\DD$ be categories and $F: \CC \to \DD$ and $G : \DD \to \CC$ functors. If $GF \Omega \xleftarrow{ \omega } \Omega$ is a final $GF$-coalgebra, then $FGF \Omega \xleftarrow{F \omega} F \Omega$ is a final $FG$-coalgebra.
\end{lemma}

By using the two lemmas above, the next theorem follows immediately.

\begin{theorem}\label{thm:reflect}
Let $\CC$ and $\DD$ be categories and $F: \CC \to \DD$ and $G : \DD \to \CC$ functors. Then
$FG$ is algebraically complete/cocomplete/compact iff $GF$ is algebraically complete/cocomplete/compact, respectively.
\end{theorem}

In order to avoid cumbersome repetition, all subsequent results are stated for
algebraic compactness. However, all results  presented in this section and the
next one (excluding Non-Example~\ref{nonexample:not-compact}) also hold true when all instances of "algebraic compactness" are
replaced with "algebraic completeness" or with "algebraic cocompleteness".

\begin{assumption}\label{ass:enrich}
Throughout the rest of the paper we assume we are given an arbitrary cartesian closed category $\VV(1, \times, \to)$ which we will use as the base of enrichment.
$\VV$-categories are written using capital calligraphic letters $(\CCE, \DDE, \ldots)$ and their underlying categories using a corresponding bold capital letter $(\CC, \DD, \ldots)$.
$\VV$-functors are also written with calligraphic letters $(\FE, \GE : \CCE \to \DDE)$ and their underlying functors using a corresponding capital letter $F, G : \CC \to \DD.$
\end{assumption}

\begin{definition}
A $\VV$-endofunctor $\TE : \CCE \to \CCE $ is \emph{algebraically compact} if its underlying endofunctor $T : \CC \to \CC$ is algebraically compact.
\end{definition}

\begin{definition}[{\cite[Definition 5.3]{fiore-plotkin}}]
A $\VV$-category $\CCE$ is \emph{$\VV$-algebraically compact} if every $\VV$-endofunctor $\TE : \CCE \to \CCE$ is algebraically compact.
\end{definition}

In particular, a $\Set$-algebraically compact category is a locally small
category $\CC$, such that every endofunctor $T : \CC \to \CC$ is algebraically
compact. In this case we simply say $\CC$ is algebraically compact.

\begin{example}
Let $\lambda$ be a cardinal and let $\Hilb$ be the category whose objects are the Hilbert spaces with dimension at most $\lambda$ and whose morphisms are the linear maps of norm at most 1. Then $\Hilb$ is algebraically compact~\cite[Theorem 3.2]{barr}.
\end{example}

For the next (very important) example, recall that a \emph{complete partial
order} (cpo) is a poset such that every increasing chain has a supremum. A cpo
is \emph{pointed} if it has a least element. A monotone map $f : X \to Y$
between two cpo's is \emph{Scott-continuous} if it preserves suprema. If, in
addition, $X$ and $Y$ are pointed and $f$ preserves the least element of $X$,
then we say that $f$ is \emph{strict}. We denote with $\cpo$ the category of
cpo's and Scott-continuous functions and we denote with $\cpobs$ the category
of pointed cpo's and strict Scott-continuous functions. The category $\cpo$
is cartesian closed, $\cpobs$ is symmetric monoidal
closed (when equipped with the smash product and strict function space) and both categories are complete and
cocomplete~\cite{abramskyjung:domaintheory}. We will see both categories as
$\cpo$-categories when equipped with the standard pointwise order on functions.

Therefore, a $\cpo$-category ($\cpobs$-category) is simply a category whose homsets have the additional structure of a (pointed) cpo and for which composition is a (strict) Scott-continuous operation in both arguments. A $\cpo$-functor ($\cpobs$-functor) then is simply a functor whose action on hom-cpo's is a (strict) Scott-continuous function.
The notion of a $\cpo$-natural transformation coincides with that of $\cpobs$-natural transformation which also coincides with the ordinary notion. Because of these reasons, it is standard in the programming semantics literature to use the same notation for $\cpo_{(\perp!)}$-enriched categorical notions and their ordinary underlying counterparts. We do the same in this paper.

\begin{example}
The category $\cpobs$ is $\cpo$-algebraically compact~\cite[Corollary 7.2.4]{fiore-thesis}.
\end{example}

Next, we show how to reflect algebraically compact $\VV$-functors.

\begin{definition}\label{def:factor}
We shall say that a $\VV$-endofunctor $\TE : \CCE \to \CCE$ has a \emph{$\VV$-algebraically compact factorisation} if
there exists a $\VV$-algebraically compact category $\DDE$ and $\VV$-functors $\FE : \CCE \to \DDE$ and $\GE : \DDE \to \CCE$ such that $\TE \cong \GE \circ \FE.$
\end{definition}

\begin{theorem}\label{thm:factor}
If a $\VV$-endofunctor $\TE : \CCE \to \CCE$ has a $\VV$-algebraically compact factorisation, then it is algebraically compact.
\end{theorem}
\begin{proof}
Taking $\DDE, \FE, \GE$ as in Definition~\ref{def:factor}, we get a $\VV$-endofunctor $\FE \circ \GE : \DDE \to \DDE$. Since $\DDE$ is $\VV$-algebraically compact, then its underlying endofunctor $F \circ G : \DD \to \DD$ is algebraically compact.
Theorem~\ref{thm:reflect} shows that $G \circ F : \CC \to \CC$ is algebraically compact. Algebraic compactness is preserved by natural isomorphisms and therefore $T \cong G \circ F$ is also algebraically compact.
\end{proof}

Using the two examples above, we easily get two corollaries.

\begin{corollary}
Any endofunctor $T: \Set \to \Set$ which factors through $\Hilb$ is algebraically compact.
\end{corollary}

\begin{corollary}\label{cor:cpo}
Any $\cpo$-endofunctor $T: \cpo \to \cpo$ which factors through $\cpobs$ via a pair of $\cpo$-functors, is algebraically compact. Thus the lifting functor $(-)_\perp : \cpo \to \cpo$ (given by freely adding a least element)
is algebraically compact.
\end{corollary}

Note that (ordinary) algebraically compact functors are \emph{not} closed under
composition. However, using the additional structure we have introduced, we can
prove the following compositionality result.

\begin{proposition}\label{prop:compose}
Let $\HE : \CCE \to \CCE$ be a $\VV$-endofunctor and $\TE: \CCE \to \CCE$ be a $\VV$-endofunctor with a $\VV$-algebraically compact factorisation. Then $\HE \circ \TE$ also has a $\VV$-algebraically compact factorisation and is thus algebraically compact.
\end{proposition}
\begin{proof}
If $\TE \cong \GE \circ \FE$, then $\HE \circ \TE \cong (\HE \circ \GE) \circ \FE$.
\end{proof}

\section{Constructive Classes of Algebraically Compact Functors}\label{sec:construct}

\begin{assumption}
Throughout the rest of the section, we assume we are given the following data. A $\VV$-category $\CCE$,
a $\VV$-algebraically compact category $\DDE$ together with $\VV$-functors $\FE : \CCE \to \DDE$ and $\GE : \DDE \to \CCE$ and a
$\VV$-endofunctor $\TE \cong \GE \circ \FE.$
\end{assumption}

Consider the following grammar:
\begin{equation}\label{eq:grammar}
A, B ::= \TE X\ |\ \HE(A_1, \ldots, A_n)\
\end{equation}
where $X$ is simply a type variable, $n$ ranges over the natural numbers (including zero) and $\HE$ ranges over $\VV$-functors $\HE : \CCE^n \to \CCE$.
Every such type expression induces a $\VV$-endofunctor $\lrb{X \vdash A} : \CCE \to \CCE,$ defined by:
\begin{align*}
\lrb{X \vdash \TE X}                 &= \TE   \\
\lrb{X \vdash \HE(A_1, \ldots, A_n)} &= \HE \circ \langle \lrb{X \vdash A_1}, \ldots, \lrb{X \vdash A_n} \rangle
\end{align*}

\begin{remark}
Since the base of enrichment $\VV$ is cartesian, tuples of $\VV$-functors, as above, are also $\VV$-functors and the above assignment is well-defined.
Also, $\VV$-algebraically compact categories have been studied only for cartesian $\VV$.
Because of these two reasons, Assumption~\ref{ass:enrich} cannot be relaxed to a symmetric monoidal closed $\VV$.
\end{remark}

\begin{theorem}\label{thm:compact-big}
Any functor $\lrb{ X \vdash A } : \CCE \to \CCE$ factors through $\FE$ and is therefore algebraically compact.
\end{theorem}
\begin{proof}
By induction. For the base case we have $\TE \cong \GE \circ \FE$. The step case is given by
\begin{align*}
\lrb{X \vdash \HE(A_1, \ldots, A_n)} &= \HE \circ \langle \lrb{X \vdash A_1}, \ldots, \lrb{X \vdash A_n} \rangle \\
  & \cong \HE \circ \langle \GE_1 \circ \FE, \ldots, \GE_n \circ \FE \rangle \\
  & = \HE \circ \langle \GE_1, \ldots, \GE_n \rangle \circ \FE,
\end{align*}
for some $\VV$-functors $\GE_i : \DDE \to \CCE$.
\end{proof}

\begin{example}
The $\VV$-functor $\TE$ is algebraically compact.
\end{example}

\begin{example}\label{ex:constant}
Any constant functor $K_c : \CC \to \CC$ is, of course, algebraically compact. This is captured by our theorem, because $K_c$ is the underlying functor of the constant $c$ $\VV$-endofunctor $\KE_c : \CCE \to \CCE$, which may be constructed using our grammar.
\end{example}

\begin{example}
If $\BBE_1 : \CCE \times \CCE \to \CCE$ and $\BBE_2 : \CCE \times \CCE \to \CCE$ are two $\VV$-bifunctors, and $\EE : \CCE \to \CCE$ is a $\VV$-endofunctor, then the endofunctors $\EE \circ \TE$ and $\BBE_1 \circ \langle \TE, \TE \rangle$
and $\BBE_2 \circ \langle \EE \circ \TE, \BBE_1 \circ \langle \TE, \TE \rangle \rangle$ are algebraically compact (among many other combinations).
\end{example}

\subsection{Special Case: Models of Mixed Linear/Non-linear Lambda Calculi}\label{sub:covariant-dill}

As a special case, our development can be applied to models of mixed linear/non-linear lambda calculi with recursive types, as we shall now explain.

In a $\cpo$-category, an \emph{embedding-projection pair} is a pair of morphisms $(e, p)$, such that $e \circ p \leq \id$ and $p \circ e = \id$.
The morphism $e$ is called an \emph{embedding} and the morphism $p$ a \emph{projection}. An \emph{e-initial object} is an initial object $0$, such that every initial map with it as source is an embedding.

\begin{definition}\label{def:model}
A model of the linear/nonlinear fixpoint calculus (LNL-FPC)~\cite{lnl-fpc} is given by the following data:
\begin{enumerate}
\item[1.] A $\cpo$-symmetric monoidal closed category $\DD$ with finite $\cpo$-coproducts, such that
$\DD$ has an e-initial object and all $\omega$-colimits over embeddings;
\item[2.] A $\cpo$-symmetric monoidal adjunction $\stikz{lnl-fpc-model.tikz}$.
\end{enumerate}
\end{definition}

In the above situation, the category $\DD$ is necessarily $\cpo$-algebraically compact, so it is an ideal setting for constructing compact algebras of $\cpo$-functors.
We will now show that the monad $T$ of this adjunction also induces a large class of algebraically compact functors on $\cpo$ (which is not $\cpo$-algebraically compact). But first, two examples of the above situation.

\begin{example}
The adjunction
\stikz{clnl-corollary.tikz},
where the left adjoint is given by domain-theoretic lifting and the right adjoint is the forgetful functor, has the required structure.
The induced monad $T : \cpo \to \cpo$ is called \emph{lifting} (see Corollary~\ref{cor:cpo}). This adjunction is in fact a computationally adequate model of LNL-FPC~\cite{lnl-fpc}.
\end{example}

\begin{example}
Let $\M$ be a small $\cpobs$-symmetric monoidal category and 
let $\widehat \M = [\M^\op, \cpobs]$ be the indicated $\cpobs$-functor category.
There exists an adjunction
\stikz{copower.tikz}, where the left adjoint is the $\cpobs$-copower with the tensor unit $I$ and the right adjoint is the representable functor (see~\cite[\S 6]{borceux:handbook2}).
Composing the two adjunctions
\stikz{big-model.tikz}
yields a LNL-FPC model.
By making suitable choices for $\M$, this data also becomes a model of Proto-Quipper-M, a quantum programming language~\cite{pqm-small} and also a model of ECLNL, a programming language for string diagrams~\cite{eclnl,eclnl2}.
\end{example}

Since $\DD$ is $\cpo$-algebraically compact, we can now construct a large class of algebraically compact functors via Theorem~\ref{thm:compact-big}.
For instance, such a subclass is given by the following corollary.
\begin{corollary}
Any endofunctor on $\cpo$ constructed using constants, $T$, $\times$ and $+,$ and such that all occurrences of the functorial variable in its definition are surrounded by $T$, is algebraically compact.
\end{corollary}

\begin{remark}
To make this more precise, one should specify a formal grammar like~\eqref{eq:grammar} to indicate the admissible functorial expressions, but it should be clear that~\eqref{eq:grammar} can be easily specialised to handle this.
\end{remark}

Next, let us consider some example endofunctors on $\cpo$.

\begin{example}\label{ex:plus}
The endofunctor $H(X) = TX + TX $ is algebraically compact. Indeed, observe that $H = + \circ \langle T, T \rangle = \lrb{X \vdash TX + TX}.$
\end{example}

\begin{example}
The endofunctor $H(X) = TX + T(TX \times TX)$ is algebraically compact. To see it, observe that
$H = + \circ \langle T, T \circ \times \circ \langle T, T \rangle \rangle = \lrb{X \vdash TX + T(TX \times TX)}.$
\end{example}

\begin{nonexample}\label{nonexample:not-compact}
The endofunctor $H(X) = X \times TX$ is not algebraically compact (its initial algebra is $\varnothing \times T\varnothing = \varnothing \xrightarrow \id \varnothing$ ). Our results do not apply to it, because the left occurrence of $X$ does not have $T$ applied to it.
For the same reason, the identity functor $\Id(X) = X$ is also not algebraically compact and not covered by our development.
\end{nonexample}

\section{Algebraically Compact Mixed-Variance Functors}
As mentioned in the introduction, algebraic compactness allows us to model
recursive datatypes which include mixed-variance functors such as function
space. In this section we show that our methods are also compatible with
recursive datatypes.

Consider a mixed-variance bifunctor $H: \CC^\op \times \CC \to \CC$. Since
$H$ is not an endofunctor, then clearly we cannot talk about $H$-algebras
or $H$-coalgebras. A more appropriate notion is that of a $H$-\emph{dialgebra},
which we will not introduce here, because of a lack of space and because the
category of $H$-dialgebras is isomorphic to the category of $\cupp H$-algebras
~\cite[\S 4]{freyd2}, where
\[ \cupp H = \langle H^\op \circ \langle \Pi_2, \Pi_1 \rangle, H \rangle : \CC^\op \times \CC \to \CC^\op \times \CC. \]
Because of this, it is standard to model recursive datatypes as endofunctors $\cupp H: \CC^\op \times \CC \to \CC^\op \times \CC$~\cite{fiore-plotkin}.

If a category $\DDE$ is $\VV$-algebraically complete, then $\DDE^\op$ is
$\VV$-algebraically cocomplete and vice versa.  Thus, $\VV$-algebraic
compactness is a self-dual notion.  Unlike the previous sections, the results
presented here do not hold for algebraically complete or cocomplete functors
and categories.

If a category $\DDE$ is $\VV$-algebraically compact in a
parameterised sense, then so is $\DDE^\op \times \DDE.$ We
omit the details of parameterised algebraic compactness, but the interested
reader may consult~\cite{fiore-plotkin}. We point out that the notions of
$\cpo$-algebraic compactness and parameterised $\cpo$-algebraic compactness
coincide~\cite[Corollary 7.2.5]{fiore-thesis} and we shall consider such a
$\cpo$-example to illustrate our methods.
But we emphasise that our methods can be adapted to the general setting
of a parameterised $\VV$-algebraically compact category $\DDE$.

Let us assume we are given an LNL-FPC model $\stikz{lnl-fpc-model.tikz}$ as in Subsection~\ref{sub:covariant-dill} with $T = G \circ F$.
In this situation, the category $\DD^\op \times \DD$ is also $\cpo$-algebraically compact and we can thus reuse Theorem~\ref{thm:factor} and Proposition~\ref{prop:compose},
where we choose the $\cpo$-algebraically compact factorisation $T^\op \times T = (G^\op \times G) \circ (F^\op \times F)$.

Consider the following grammar:
\begin{equation}\label{eq:mixed-grammar}
A, B ::= c\ |\ T X\ |\ H A\ |\ A+B\ |\ A \times B\ |\ A \to B,
\end{equation}
where $c$ ranges over the objects of $\cpo$ and $H$ ranges over $\cpo$-endofunctors on $\cpo$.
Every such type expression induces a $\cpo$-endofunctor $\lrb{X \vdash A} : \cpo^\op \times \cpo \to \cpo^\op \times \cpo,$ defined by:
\begin{align*}
\lrb{X \vdash T X}                 &= T^\op \times T   \\
\lrb{X \vdash c}                   &= K_{(c,c)}        \\
\lrb{X \vdash H A}                 &= (H^\op \times H) \circ \lrb{X \vdash A} \\
\lrb{X \vdash A+B} &= \left( + \circ \langle \Pi_2 \lrb{X \vdash A}, \Pi_2 \lrb{X \vdash B} \rangle \right)^\smallsmile \\
\lrb{X \vdash A \times B} &= \left( \times \circ \langle \Pi_2 \lrb{X \vdash A}, \Pi_2 \lrb{X \vdash B} \rangle \right)^\smallsmile \\
\lrb{X \vdash A \to B} &= \left( [ - \to - ] \circ \langle \Pi_1 \lrb{X \vdash A}, \Pi_2 \lrb{X \vdash B} \rangle \right)^\smallsmile,
\end{align*}
where $K_{(c,c)}$ is the constant $(c,c)$ endofunctor on $\cpo^\op \times \cpo$ and $[- \to - ] : \cpo^\op \times \cpo \to \cpo$ is the internal-hom.
\begin{remark}
The last three cases in the above assignment are essentially the same as the standard interpretation of types within FPC~\cite[Definition 6.2]{fiore-plotkin}.
\end{remark}

\begin{theorem}
Every functor $\lrb{X \vdash A} : \cpo^\op \times \cpo \to \cpo^\op \times \cpo$ factors through $F^\op \times F$ and is therefore algebraically compact.
\end{theorem}
\begin{proof}
Simple proof by induction. The first three cases are obvious. For the last three cases, simply use the fact that $(H \circ (F^\op \times F))^\smallsmile = \cupp H \circ (F^\op \times F),$
which can be proved after recognising that $(-)^\op$ is a \emph{covariant} operation with respect to functor composition.
\end{proof}

\begin{example}
Consider the functor $H(X, Y) = [TX \to TY] : \cpo^\op \times \cpo \to \cpo$.
Then the functor $\cupp H : \cpo^\op \times \cpo \to \cpo^\op \times \cpo$ is algebraically compact,
because:
\begin{align*}
 \cupp H = ([- \to - ] \circ (T^\op \times T))^\smallsmile &= ([ - \to - ] \circ \langle \Pi_1, \Pi_2 \rangle \circ (T^\op \times T))^\smallsmile \\
 &= \left( [ - \to - ] \circ \langle \Pi_1 \circ (T^\op \times T), \Pi_2 \circ (T^\op \times T) \rangle \right)^\smallsmile \\
 &= \lrb{X \vdash TX \to TX}.
\end{align*}
\end{example}

\begin{nonexample}
Consider the internal-hom functor $[- \to - ] : \cpo^\op \times \cpo \to \cpo$.
Then $[ - \cupp \to - ]$ is not algebraically compact, because its initial algebra is given by
\[ \left( [ - \cupp \to -] (1, \varnothing) = ([ \varnothing \to 1 ] , [1 \to \varnothing] ) = (1, \varnothing) \right) \xrightarrow{\id} (1, \varnothing), \]
which is not its final coalgebra.
Our results do not apply to $[ - \cupp \to -]$, because $T$ does not occur anywhere in its definition.
\end{nonexample}

\section{Comparison with Limit-Colimit Coincidence Results}

The focus in this paper is to study algebraically compact endofunctors on
categories which do not necessarily have a zero object. In~\cite{barr} Michael Barr
considers this situation and he
presents a more general version of the standard limit-colimit coincidence theorem~\cite{smyth-plotkin}.
The increased generality allows him to establish the existence of algebraically compact endofunctors on categories that do not have a zero object.
In this section, we will compare his results about $\cpo$-categories with ours.

\begin{theorem}[{\cite[Theorem 5.4]{barr}}]\label{thm:barr}
Let $\CC$ be a $\cpo$-category with initial object $\varnothing$ and terminal object $1$. Assume further $\CC$ has colimits of initial sequences of $\cpo$-endofunctors.
Then the class of $\cpo$-endofunctors for which there is a morphism $l : 1 \to H\varnothing$ such that
$\left( H1 \xrightarrow{} 1 \xrightarrow l H\varnothing \xrightarrow{Hh} H1 \right) \leq \id_{H1},$ where $h: \varnothing \to 1$ is the unique arrow, is algebraically compact.
\end{theorem}

First, a necessary condition in the above situation.

\begin{proposition}
In the situation of Theorem~\ref{thm:barr}, the hom-cpo $\CC(H1, H1)$ is pointed.
\end{proposition}
\begin{proof}
Let $\perp\ = \left( H1 \xrightarrow{} 1 \xrightarrow l H\varnothing \xrightarrow{Hh} H1 \right)$. Let $f: H1 \to H1$ be an arbitrary morphism. Then
\[ \perp\ =\ \perp \circ f  \leq \id \circ f = f. \]
\end{proof}

We may now see that Barr's theorem does not behave well when dealing with constant functors or with functors involving coproducts.

\begin{example}
Consider the constant functor $K_2 : \cpo \to \cpo$ where $2$ is any two point
cpo equipped with the discrete order. As we explained in
Example~\ref{ex:constant}, our development captures the fact that $K_2$ is
algebraically compact. However, Barr's theorem does not show this,
because $\cpo(2, 2)$ is not pointed.
\end{example}

\begin{example}
Consider the functor $H(X) = X_\perp + X_\perp : \cpo \to \cpo$ where $(-)_\perp$ is given by lifting.
Our development showed in
Example~\ref{ex:plus} that this functor is algebraically compact. However, Barr's theorem does not show this,
because $\cpo(1_\perp + 1_\perp, 1_\perp + 1_\perp)$ is not pointed.
\end{example}

A natural question to ask is whether there exists an algebraically compact
functor described by Theorem~\ref{thm:barr}, but not captured by the methods
presented here. We leave this for future work.

We also provided a compositionality result
(Proposition~\ref{prop:compose}) which then allowed us to present a
\emph{constructive} description of large classes of algebraically compact
functors (Section~\ref{sec:construct}). So far this has not been done with Barr's results.

\section{Related Work}

The solution of recursive domain equations is based on the construction of
initial algebras \cite{adamek-original,lehman-smyth} and on the construction of
compact algebras when mixed-variance functors are involved
\cite{fiore-thesis,fiore-plotkin,freyd2,freyd}.  The limit-colimit coincidence theorem
\cite{smyth-plotkin} for $\cpo$-enriched categories with sufficient structure
is perhaps the most common way of constructing such compact algebras.  In this
paper we have focused on the construction of compact algebras within categories
with little structure that does not admit utilising the above mentioned
approaches.  Our motivation for doing this is to consider denotational
interpretations of mixed linear/non-linear recursive types.

Another approach for modelling mixed linear/non-linear recursive types is
described in~\cite{lnl-fpc,lnl-fpc-lmcs} where the authors interpret non-linear
types within a carefully constructed subcategory of $\cpo$. That method works
only for $\cpo$-categories whereas the techniques presented here work
for arbitrary $\VV$-categories. Also, the set of type expressions that can be
interpreted with the methods from~\cite{lnl-fpc,lnl-fpc-lmcs} is incomparable with the one
presented here (neither is a subset of the other). However, the main idea in
\cite{lnl-fpc,lnl-fpc-lmcs} is to reflect the initial algebra structure from
certain (sub)categories and not necessarily the compact algebra structure.
This method has found further applications in constructing denotational models
for quantum programming \cite{qpl-tcs,qpl-fossacs} and for affine type systems
\cite{cmcs}.

\section{Conclusion}

We established new results about algebraically compact functors without relying
on limits, colimits or their coincidence. We arrived at these results in a more
abstract way by observing that any enriched endofunctor is algebraically
compact, provided that it factors through a category which is algebraically
compact in an enriched sense. This then allowed us to establish large classes
of algebraically compact functors which also admit a constructive description. Our results are
compositional and nicely complement other existing approaches in the literature
which do rely on a limit-colimit coincidence.

\paragraph{Acknowledgements.} The author wishes to thank the anonymous
reviewers for their feedback and he gratefully acknowledges financial support
from the French projects ANR-17-CE25-0009 SoftQPro and PIA-GDN/Quantex.

\bibliography{refs}

\end{document}